\documentclass[pdflatex,sn-mathphys-num]{sn-jnl}

\usepackage{graphicx}%
\usepackage{multirow}%
\usepackage{amsmath,amssymb,amsfonts}%
\usepackage{amsthm}%
\usepackage{amsmath}
\usepackage{mathrsfs}%
\usepackage[title]{appendix}%
\usepackage{xcolor}%
\usepackage{textcomp}%
\usepackage{manyfoot}%
\usepackage{booktabs}%
\usepackage{algorithm}%
\usepackage{algorithmicx}%
\usepackage{algpseudocode}%
\usepackage{listings}%
\usepackage{refcount}%
\usepackage{hyperref}
\usepackage{cleveref}

\newif\ifhighlight
\highlightfalse 

\usepackage{xcolor}
\definecolor{temp}{rgb}{0.0, 0.3, 0.5} 
\definecolor{mygray}{gray}{0.6}
\definecolor{mhgray}{gray}{0.4}
\definecolor{lgray}{gray}{0.88} 
\definecolor{pakistan}{rgb}{0.15, 0.55, 0.0}
\definecolor{kimidori}{rgb}{0.7,0.95,0.3}
\definecolor{palebrown}{rgb}{0.6, 0.46, 0.33}
\definecolor{mypink}{rgb}{0.9, 0.0, 0.4} 
\definecolor{antiquefuchsia}{rgb}{0.57, 0.36, 0.51}
\definecolor{darkcyan}{rgb}{0.0, 0.55, 0.55}
\definecolor{teal}{rgb}{0.0, 0.5, 0.5}
\definecolor{momo}{rgb}{1.0,0.7,0.8} 
\definecolor{yamabuki}{rgb}{1.0, 0.86, 0.0} 
\definecolor{navy}{rgb}{0.0,0.0,0.9}
\definecolor{wine}{rgb}{0.6,0.0,0.0}
\ifhighlight
  \newcommand{\mypink}[1]{\textcolor{mypink}{#1}}

\else
  \newcommand{\mypink}[1]{#1}

\fi

\theoremstyle{thmstyleone}%
\newtheorem{theorem}{Theorem}
\newtheorem{proposition}[theorem]{Proposition}%

\newtheorem{lemma}[theorem]{Lemma}%

\theoremstyle{thmstyletwo}%

\theoremstyle{thmstylethree}%

\newtheorem{problem}{Problem}
\usepackage{booktabs}
\usepackage{algorithm}
\usepackage{algpseudocode}
\algrenewcommand\Require{\item[\textbf{Input:}]}
\algrenewcommand\Ensure{\item[\textbf{Output:}]}

\newcommand{\argmax}{\operatorname{arg\,max}}

\makeatletter
\newcounter{experimentcounter}
\newcommand{\experimentlabel}[1]{%
    \refstepcounter{experimentcounter}%
    \def\@currentlabel{\theexperimentcounter}%
    \label{#1}%
}
\makeatother

\begin{document}

\title[Orientability of Undirected Phylogenetic Networks to a Desired Class]{Orientability of Undirected Phylogenetic Networks to a Desired Class: Practical Algorithms and Application to Tree-Child Orientation}


\author[1]{\fnm{Tsuyoshi} \sur{Urata}}\email{uratsuyo244@moegi.waseda.jp}

\author[1]{\fnm{Manato} \sur{Yokoyama}}\email{mana.aki.aya@akane.waseda.jp}

\author[1]{\fnm{Haruki} \sur{Miyaji}}\email{h.miyaji@ruri.waseda.jp}

\author*[2]{\fnm{Momoko} \sur{Hayamizu}}\email{hayamizu@waseda.jp}

\affil[1]{\orgdiv{Department of Pure and Applied Mathematics}, \orgname{Graduate School of Fundamental Science and Engineering, Waseda University}, \country{Japan}}

\affil[2]{\orgdiv{Department of Applied Mathematics, Faculty of Science and Engineering}, \orgname{Waseda University}, \country{Japan}}


\abstract{The $\mathcal{C}$-\textsc{Orientation} problem asks whether it is possible to orient an undirected graph to a directed phylogenetic network of a desired network class $\mathcal{C}$. This problem arises, for example, when visualising evolutionary data, as popular methods such as Neighbor-Net are distance-based and inevitably produce undirected graphs. The complexity of $\mathcal{C}$-\textsc{Orientation} remains open for many classes $\mathcal{C}$, including binary tree-child networks, and practical methods are still lacking. 
In this paper, we propose \mypink{(1)} an exact FPT algorithm for $\mathcal{C}$-\textsc{Orientation}, applicable to any class $\mathcal{C}$ \mypink{admitting a tractable membership test}, and parameterised by the reticulation number and the maximum size of minimal basic cycles\mypink{, and (2)} a very fast heuristic for \textsc{Tree-Child Orientation}.
While the state-of-the-art for $\mathcal{C}$-\textsc{Orientation} is a simple exponential time algorithm whose computational bottleneck lies in searching for appropriate reticulation vertex placements, our methods significantly reduce this search space. Experiments show that, although our FPT algorithm is still exponential, it significantly outperforms the existing method. The heuristic runs even faster but with increasing false negatives as the reticulation number grows. Given this trade-off, we also discuss theoretical directions for improvement and biological applicability of the heuristic approach.}




\maketitle

\section{Introduction}\label{sec1}

Phylogenetic networks are a powerful tool for representing complex evolutionary relationships between species that cannot be adequately modelled by trees. These networks are particularly useful in the presence of reticulate events, such as hybridisation, horizontal gene transfer (HGT) and recombination. Hybridisation refers to the interbreeding of individuals from different species, which can lead to the formation of a new hybrid species that shares genetic material from both parent species. HGT is the transmission of copied genetic material to another organism without being its offspring, a process particularly common in bacteria and archaea. Recombination is the exchange of genetic material between different genomes and is a common occurrence not only in viruses but also in bacteria, eukaryotes, and other organisms. While phylogenetic trees depict the hierarchical branching of evolutionary lineages, networks can represent both the hierarchical and non-hierarchical connections resulting from reticulate events, providing a more nuanced view of evolutionary history.

However, constructing directed phylogenetic networks from biological data remains a challenging task. 
Distance-based methods, such as Neighbor-Net \mypink{\cite{bryant2004neighbor}}, are widely used because they are scalable and helpful in visualising the data. 
However, the resulting networks are inevitably undirected, often making it difficult to interpret the evolutionary history. To provide a more informative representation of the data, it is meaningful to develop a method for transforming undirected graphs into directed phylogenetic networks in a way that ensures the resulting network has a desired property.

Recently, Huber {\textit{et al.}}\,\cite{huber2024orienting} introduced two different orientation problems, each considering orientation under a different constraint. The \textsc{Constrained Orientation} problem asks whether a given undirected phylogenetic network can be oriented to a directed network under the constraint of a given root edge and desired in-degrees of all vertices, and asks to find a feasible orientation under that constraint if one exists. In \cite{huber2024orienting}, it was shown that such a feasible orientation is unique if one exists, and a linear time algorithm for solving this problem was provided.
However, when an undirected network has been created from data, it is not usually the case that there is complete knowledge of where to insert the root and which vertices are reticulations. The $\mathcal{C}$-\textsc{Orientation} problem does not constrain the position of the root or the in-degree of the vertices. Instead, it asks whether a given binary network can be oriented to a directed phylogenetic network belonging to a desired class $\mathcal{C}$. 
 
The complexity of $\mathcal{C}$-\textsc{Orientation} is not fully understood for many classes $\mathcal{C}$, and no study has discussed practically useful methods. In \cite{huber2024orienting}, \textsc{Tree-Based Orientation} was shown to be NP-hard. Maeda {\textit{et al.}}\,\cite{maeda2023orienting} conjectured that \textsc{Tree-Child Orientation} is NP-hard. 
 \mypink{Given a graph with maximum degree five, determining whether or not the graph is tree-child orientable was shown to be NP-hard \cite{bulteau2023turning, docker2025}.}
\mypink{Even for} graphs of maximum degree three (i.e. binary networks), determining the orientability to a desired class, such as tree-child, tree-based or reticulation-visible networks, may not be easy \cite{bulteau2023turning}.
Although \cite{huber2024orienting} provided FPT algorithms for a special case of $\mathcal{C}$-\textsc{Orientation} where $\mathcal{C}$ satisfies several conditions, which are  theoretically applicable to various classes including tree-child networks but are not easy to implement,  there is no FPT algorithm to solve $\mathcal{C}$-\textsc{Orientation} in its general form. Also, no studies have pursued practically useful heuristics.
 
In this paper, we provide a practically efficient exponential time algorithm for $\mathcal{C}$-\textsc{Orientation} (Algorithm \ref{alg:proposed.BF}), which is FPT in both the reticulation number and the maximum size of minimal basic cycles selected by the algorithm. We also present a heuristic method for \textsc{Tree-Child Orientation} (Algorithm \ref{alg:proposed.fast}) which, although still exponential, runs very fast in practice because it only considers reticulation placements to maximise the sum of their pairwise distances. Using artificially generated networks, we compare the accuracy and execution time of the proposed methods for solving \textsc{Tree-Child Orientation} with those of the existing exponential time algorithm for $\mathcal{C}$-\textsc{Orientation} (Algorithm 2 in \cite{huber2024orienting}). Our theoretical and empirical results demonstrate the usefulness of Algorithm \ref{alg:proposed.BF}, especially for relatively large input graphs with $5$ or more reticulations, where the exponential time method in \cite{huber2024orienting} becomes computationally infeasible. Our Algorithm \ref{alg:proposed.fast}, while much faster than Algorithm \ref{alg:proposed.BF}, tends to decrease in accuracy as the reticulation number increases.

The rest of the paper is structured as follows. 
Section \ref{sec:definitions.notation} provides the necessary mathematical definitions and notation, including the definition of phylogenetic networks. 
Section \ref{sec:problems.literature} briefly reviews relevant results from \cite{huber2024orienting} and formally states the problems of interest. Section \ref{sec:theoretical.aspects} gives the theoretical background of our proposed methods, including the concept of `cycle basis' and a theorem that allows us to reduce the search space (Theorem \ref{thm:R.cycle}). Section \ref{sec:proposed.methods} describes the proposed exact method (Algorithm \ref{alg:proposed.BF}) and a heuristic method (Algorithm \ref{alg:proposed.fast}). Theorem \ref{thm:exact} ensures that the heuristic is correct when $r\leq 2$. 
We analyse the time complexity of Algorithm \ref{alg:proposed.BF}.
Section \ref{sec:experiments} explains the experimental setup, including the method used to create undirected graphs (details can be found in the appendix), and presents the results of three  experiments. Section \ref{sec:discussion} discusses the limitations and biological application of the heuristic method. Finally, Section \ref{sec:conclusion.future.work} concludes the paper and outlines future research directions.

\section{Definitions and Notation}\label{sec:definitions.notation}
\subsection{Graph theory}
An \emph{undirected graph} is an ordered pair $(V, E)$ consisting of a set $V$ of \emph{vertices} and a set $E$ of \emph{edges} between vertices without any orientation. Given an undirected graph $G$, its vertex set and edge set are denoted by $V(G)$ and $E(G)$, respectively. 
An edge of an undirected graph between vertices $u$ and $v$ is denoted by $\{u, v\}$ or $\{v, u\}$.  
An undirected graph is \emph{simple} if it contains neither a loop nor multiple edges, namely, any edge $\{u, v\}$ satisfies $u \neq v$, and any edges $\{u, v\} \neq \{u^\prime, v^\prime\}$ satisfy at least one of $u \neq u^\prime$ and $v \neq v^\prime$.    
Given simple graphs $G_i = (V_i,E_i)$ ($i=1,2,\dots, n$), the graph $(V, E)$ with $V:=\bigcup_{i=1}^nV_i$ and  $E:=\bigcup_{i=1}^nE_i$ is called \emph{the union} of the graphs $G_i$. 

For a vertex $v$ of an undirected graph $G$, the \emph{degree} of $v$ in $G$, denoted by $\mathrm{deg}_G(v)$, is the number of edges of $G$ that joins $v$ with another vertex of $G$.

An \emph{(undirected) path} is an undirected graph $P$ with a vertex set $\{v_1, v_2, \ldots, v_\ell\}$ and an edge set $\{\{v_1, v_2\}, \{v_2, v_3\}, \ldots, \{v_{\ell -1}, v_\ell \}\}$. The number of edges of a path $P$ is called the \emph{length} of  $P$. 
An \emph{(undirected) cycle} is an undirected graph $C$ with a vertex set $\{v_1, v_2, \ldots, v_\ell\}$ and an edge set $E = \{\{v_1, v_2\}, \{v_2, v_3\}, \ldots, \{v_{\ell -1}, v_\ell \}, \{v_\ell, v_1\}\}$. The number of edges of a cycle $C$ is called the \emph{length} of $C$.
A \emph{subgraph} of an undirected graph $G=(V, E)$ is a graph $G^\prime=(V^\prime, E^\prime)$ such that $V^\prime \subseteq V$ and $E^\prime \subseteq E$. In this case, $G$ \emph{contains} $G^\prime$. 
If an undirected graph $G$ contains a cycle $C$ as a subgraph, $C$ is a cycle \emph{of} $G$. 
An undirected graph $G$ is \emph{connected} if for any $u,v\in V(G)$, $G$ contains a path between $u$ and $v$. The \emph{distance} $d_G (u,v)$ between $u$ and $v$ in $G$ is defined by the length of the shortest path connecting $u$ and $v$ in $G$. 

A \emph{directed graph} is an ordered pair $(V, A)$ consisting of a set $V$ of vertices and a set $A$ of oriented edges called \emph{arcs}. 
Given a directed graph $D$, its vertex set and arc set are denoted by $V(D)$ and $A(D)$, respectively. 
An arc that goes from vertex $u$ to vertex $v$ is denoted by $(u, v)$. 
Given an arc $(u,v)$, $u$ is a \emph{parent} of $v$ and $v$ is a \emph{child} of $u$. 
A directed graph is \emph{simple} if it contains neither a loop nor multiple arcs.
 For a vertex $v$ of a directed graph $D$, the \emph{in-degree} of $v$ in $D$, denoted by $\mathrm{indeg}_D(v)$,  is the number of arcs of $D$ that arrive at $v$. 
 Likewise, the \emph{out-degree} of $v$ in $D$, denoted by $\mathrm{outdeg}_D(v)$, is the number of arcs of $D$ that start from $v$. 
A \emph{directed cycle} is a directed graph with a vertex set $V = \{v_1, v_2, \ldots, v_\ell \}$ and an arc set $A = \{(v_1, v_2), (v_2, v_3), \ldots, (v_{\ell-1}, v_\ell), (v_\ell, v_1)\}$. 
A subgraph of a directed graph is defined in the same way as before. A \emph{directed acyclic graph (DAG)} is a directed graph that contains no cycle as its subgraph. 
Given a directed graph $D$, the undirected graph obtained by ignoring the direction of all its arcs is called \emph{the underlying graph} of $D$ and  denoted by $U(D)$.

\subsection{Phylogenetic networks}
Throughout this paper, $X$ is a finite set with $|X| \ge 2$, representing a set of the present-day species of interest.
All graphs considered here are simple and finite, meaning the numbers of vertices and edges are finite.
An \emph{undirected binary phylogenetic network}  on $X$ is a simple, connected, undirected graph $N$ such that its vertex set $V$ is partitioned into  $V_I :=\{v\in V \mid \mathrm{deg}_N(v)=3\}$ and $V_L:=\{v\in V \mid \mathrm{deg}_N(v)=1\}$, and $V_L$ can be identified with $X$. Each vertex in $V_I$ and in $V_L$ is called an \emph{internal vertex} and a \emph{leaf} of $N$, respectively.  

A \emph{directed binary phylogenetic network} on $X$ is a simple, acyclic directed graph $D$ such that the underlying graph of $D$ is connected, the vertex set $V$ of $D$ contains a unique vertex $\rho$ with $({\mathrm{indeg}}_D(\rho), \mathrm{outdeg}_D(\rho))=(0,2)$ and the set $V\setminus \{\rho\}$ is partitioned into $V_T:=\{v\in V \mid (\mathrm{indeg}_D(v), \mathrm{outdeg}_D(v))=(1,2)\}$ and $V_R:=\{v\in V \mid (\mathrm{indeg}_D(v), \mathrm{outdeg}_D(v))=(2,1)\}$, and $V_L:=\{v\in V \mid (\mathrm{indeg}_D(v), \mathrm{outdeg}_D(v))=(1,0)\}$ that can be identified with $X$. The vertex $\rho$ is called \emph{the root} of $D$, and each vertex in $V_T$, in $V_R$ and in $V_L$ is called a \emph{tree vertex}, a \emph{reticulation} and a \emph{leaf}  of $D$, respectively. 

A directed binary phylogenetic network $D$ on $X$ is a \emph{tree-child network} if every non-leaf vertex of $D$ has at least one tree vertex as a child \cite{cardona2008comparison}. Tree-child networks are characterised by the absence of the two forbidden subgraphs \cite{semple2016phylogenetic} that are illustrated in Fig. \ref{fig:forbidden}. Namely, tree-child networks contain neither a vertex with two reticulation children nor a reticulation with a reticulation child.  

\begin{figure}[h]
\begin{center}
\includegraphics[width=.25\textwidth]{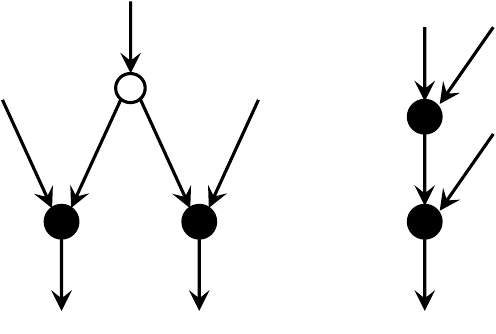}
\caption{The forbidden subgraphs for tree-child networks. Left: a vertex with two reticulation children. Right: a reticulation with a reticulation child. The black vertices represent reticulations and the white circle indicates a tree vertex.}
\label{fig:forbidden}
\end{center}
\end{figure}

\section{Acyclic orientation problems and known results}\label{sec:problems.literature}
For an undirected phylogenetic network $N$ on $X$, \emph{orienting} $N$ is a procedure that inserts the root $\rho$ into a unique edge $e_\rho = \{u, v\}$ of $N$ by replacing the edge $\{u, v\}$ with the arcs $(\rho, u), (\rho, v)$ and then orienting the other edges so that the resulting graph $\vec{N}$ is a directed phylogenetic network on $X$. In other words, a directed phylogenetic network $\vec{N}$ on $X$ is an \emph{orientation} of $N$ if the underlying graph $U(\vec{N})$ becomes isomorphic to $N$ after suppressing the root $\rho$ of $U(\vec{N})$ (i.e. replacing the undirected path $u-\rho-v$ with the edge $\{u,v\}$).

\mypink{Determining whether an undirected graph can be oriented into a directed graph satisfying certain constraints is a long-standing topic in graph theory, with classical foundations \cite{robbins1939theorem, Nash-Williams_1960} and ongoing developments (e.g. \cite{buchinoriented, hartmann2023make, disser2016degree, asahiro2007graph}).  
However, in the context of phylogenetic networks, such orientation problems have only recently been studied \cite{huber2024orienting, maeda2023orienting, bulteau2023turning, docker2025}}.  
Huber {\textit{et al.}}\,\cite{huber2024orienting} discussed two types of problems,  \textsc{Constrained Orientation} and $\mathcal{C}$-\textsc{Orientation}.

\subsection{Orientation with a desired root position and in-degrees}
\begin{problem}(\textsc{Constrained Orientation})
\label{prob:constrained orientation}
\begin{description}
\item[INPUT: ]
An undirected (not necessarily binary) phylogenetic network $N = (V, E)$ on $X$, an edge $e_{\rho} \in E$ into which a unique root $\rho$ is inserted, and the desired in-degree $\delta_N^-(v)$ of each $v \in V$. 
\item[OUTPUT: ]
An orientation $\vec{N}$ of $N$ that satisfies the constraint $(e_{\rho}, \delta_N^-)$ if it exists, and `NO' otherwise.
\end{description}
\end{problem}

We note that when $N$ is binary, the constraint $\delta_N^-$ in Problem~\ref{prob:constrained orientation} specifies which internal vertices of $N$ are to become reticulations or tree vertices in $\vec{N}$. 
Below we restate the relevant result from \cite{huber2024orienting}.

\begin{theorem}[Part of Theorems 1 and 2 in \cite{huber2024orienting}]\label{thm:huber.constrained}
In Problem \ref{prob:constrained orientation}, 
if there exists an orientation $\vec{N}$ of $N$ that satisfies the constraint $(e_\rho, \delta_N^-)$, then $\vec{N}$ is unique for $(e_\rho, \delta_N^-)$.  
Algorithm 1 in \cite{huber2024orienting} can determine whether or not $\vec{N}$ exists and find $\vec{N}$ if it does, both in $O(|E|)$ time.
\end{theorem}

Although Theorem~\ref{thm:huber.constrained} is sufficient for our purpose, a necessary and sufficient condition for when $\vec{N}$ exists under the constraint $(e_{\rho}, \delta_N^-)$ was provided in \cite{huber2024orienting} using the notion of `degree cut'. The interested reader is referred to  \cite{huber2024orienting}.

\subsection{Orientation to a desired class $\mathcal{C}$ of networks}
The next problem is about orientation under a different constraint. It asks whether a given graph can be oriented to be a directed phylogenetic network in a desired class $\mathcal{C}$, where the position of the root edge and the in-degree of each vertex are unknown. In \cite{huber2024orienting}, $\mathcal{C}$-\textsc{Orientation} was defined under the assumption that the input $N$ is binary, unlike Problem~\ref{prob:constrained orientation}. \mypink{Note that the class $\mathcal{C}$ is assumed to admit a tractable membership test.}

\begin{problem}($\mathcal{C}$-\textsc{Orientation})
    \label{prob:C-orientation}
    \begin{description} 
        \item[INPUT: ] 
        An undirected binary phylogenetic network $N$ on $X$.
        \item[OUTPUT: ]
        An orientation $\vec{N}$ of $N$ such that $\vec{N}$ belongs to the class $\mathcal{C}$ of directed binary phylogenetic networks on $X$ if it exists, and `NO' otherwise.
        \end{description}
\end{problem}

Huber {\textit{et al.}}\,\cite{huber2024orienting} described a simple exponential time algorithm for solving Problem~\ref{prob:C-orientation} (Algorithm 2 in \cite{huber2024orienting}), which uses the above-mentioned $O(|E|)$ time algorithm for Problem \ref{prob:constrained orientation}. It repeatedly solves Problem \ref{prob:constrained orientation} for all possible combinations of the root edge $e_\rho \in E$ and a set  $\{v_1,\dots,v_r\}=V_R\subseteq V$ of reticulations (i.e. vertices of desired in-degree $2$) until it finds an orientation $\vec{N}$ of $N$ that satisfies $(e_\rho, \delta_N^-)$ and belongs to class $\mathcal{C}$. 

The complexity of $\mathcal{C}$-\textsc{Orientation} depends on $\mathcal{C}$ but is still unknown for most of the popular classes of phylogenetic networks. For example, when $\mathcal{C}$ is the class of trees, Problem \ref{prob:C-orientation} is obviously solvable in polynomial time, and when $\mathcal{C}$ is the class of tree-based networks, it was shown in \cite{huber2024orienting} that the problem is NP-hard. An important remark is that if $\mathcal{C}^\prime$ is a subclass of $\mathcal{C}$, it does not necessarily imply that $\mathcal{C}^\prime$-\textsc{Orientation} is easier or harder than $\mathcal{C}$-\textsc{Orientation}. In fact, the complexity of the following problem is still open \cite{docker2025, bulteau2023turning}.

\begin{problem}(\textsc{Tree-Child Orientation})
    \label{prob:tree-child orientation}
    \begin{description} 
        \item[INPUT: ] 
        An undirected binary phylogenetic network $N$ on $X$.
        \item[OUTPUT: ]
        An orientation $\vec{N}$ of $N$ that is a tree-child network on $X$ if it exists, and `NO' otherwise.
        \end{description}
\end{problem}

While FPT algorithms for a special case of $\mathcal{C}$-\textsc{Orientation} were provided in \cite{huber2024orienting}, it remains challenging to develop 
a practical method for this type of orientation problems. This motivates us to explore a heuristic approach for solving Problem \ref{prob:tree-child orientation}. 

In what follows, when there exists such an orientation $\vec{N}$ of $N$ as described in Problem \ref{prob:tree-child orientation}, we say that $N$ is \emph{tree-child orientable} and call $\vec{N}$ a \emph{tree-child orientation} of $N$.

\section{Theoretical aspects of the proposed methods}\label{sec:theoretical.aspects}
Before discussing Problem \ref{prob:tree-child orientation}, we consider a general setting where we want to orient any undirected phylogenetic network $N=(V,E)$ on $X$ to a rooted directed phylogenetic network $\vec{N}$ on $X$. Then $\vec{N}$ must contain $r = |E| - |V| + 1$ reticulation vertices (which can be easily verified by induction, and can also be derived from the equations in Lemma 2.1 in \cite{mcdiarmid2015counting}), and so we need to decide which $r$ vertices among $|V|-|X|$ internal vertices of $N$ will have in-degree 2 in $\vec{N}$.
The number of possible ways to select $r$ vertices from non-leaf vertices of $N$ is $\binom{|V|-|X|}{r}$, which is exponential. We will now consider how we can reduce the number of candidates to examine.

\begin{lemma}
	\label{lem:R.position}
   Let $N=(V,A)$ be any directed acyclic graph and let $U(N)=(V, E)$ be the underlying graph of $N$. Then, any cycle $C$ of $U(N)$ has a vertex whose in-degree in $N$ is at least $2$.
\end{lemma}

\begin{proof}
    To obtain a contradiction, suppose $U(N)$ contains a cycle $C=(v_0,v_1,\dots,v_{n-1})$ such that \mypink{the desired in-degree} $\mathrm{indeg}_N(v_i)\leq 1$ for every $v_i \in V(C)$. Let $\vec{C}$ be the subgraph of $N$ that corresponds to $C$. 
\mypink{We claim that each $v_i \in V(C)$ satisfies  $\mathrm{indeg}_N(v_i) = 1$. To see this, assume that there exists $v_i \in V(C)$ with $\mathrm{indeg}_N(v_i) = 0$. As $C$ is a cycle, $|V(\vec{C})| = |A(\vec{C})| = n$, so $\sum_{v \in V(\vec{C})} \mathrm{indeg}_{\vec{C}}(v)  = n$. By $\mathrm{indeg}_N(v_i) = 0$, we have $\mathrm{indeg}_{\vec{C}}(v_i)=0$, so $\sum_{v \in V(\vec{C})} \mathrm{indeg}_{\vec{C}}(v) = \sum_{v \in V(\vec{C})\setminus \{v_i\}} \mathrm{indeg}_{\vec{C}}(v)$. For any $v \in V(\vec{C})$, we have $\mathrm{indeg}_{\vec{C}}(v) \leq \mathrm{indeg}_N(v) \leq 1$. Thus, $\sum_{v \in V(\vec{C})\setminus \{v_i\}} \mathrm{indeg}_{\vec{C}}(v) \leq n-1$, which contradicts $\sum_{v \in V(\vec{C})} \mathrm{indeg}_{\vec{C}}(v) = n$. Hence, the above claim holds. }
    Let $\{v_{i-1}, v_i\}$ and $\{v_i, v_{i+1}\}$ be two consecutive undirected edges of $C$. We may assume that $\vec{C}$ has the arc $(v_{i-1}, v_i)$, instead of $(v_i, v_{i-1})$. Since $\mathrm{indeg}_N(v_i) = 1$, $\vec{C}$ contains $(v_i, v_{i+1})$, not $(v_{i+1}, v_i)$. The same argument applies to all arcs of $\vec{C}$. It follows that $\vec{C}$ is a directed cycle, a contradiction. 
\end{proof}

Lemma~\ref{lem:R.position} allows us to exclude some of the inappropriate reticulation placements (specifically, cases where there is a cycle having no reticulation) that make orientation impossible. However, checking all cycles in a graph is computationally inefficient. To reduce the search space, we will now introduce some relevant concepts.

For a connected undirected graph $N=(V,E)$, the \emph{cycle rank} of $N$ is defined to be the number $r:=|E|-|V|+1$ (e.g. p.24 in \cite{gross2003handbook}), which is also known as circuit rank, cyclomatic number and the (first-order) Betti number of $N$. Note that $r$ is zero if $N$ is a tree and that $r$ is the number of desired reticulation vertices if $N$ is a binary phylogenetic network that is an instance of Problem \ref{prob:constrained orientation}. 
The cycle rank $r$ of $N$ can also be interpreted as the rank of a vector space called the `cycle space', where each of the $r$ basis vectors, \mypink{each} called \emph{basic cycles}, corresponds to the edge-set of a simple cycle in $N$. 
If we define the summation of cycles $C$ and $C^\prime$ as the \mypink{even-degree graph} induced by the symmetric difference of their edge-sets, then the cycle space of $N$ can be identified with the set of even-degree subgraphs of $N$, so any cycle in $N$ can be expressed as a sum of basic cycles in a cycle basis $\mathcal{S}$ (for details, see Section 1.9 of \cite{Diestel2012GraphT4}, Section 4.3 of \cite{bondy2008graph} and Section 6.4.2 of \cite{gross2003handbook}).

Proposition \ref{prop:digraph} will be used in the proof of the main theorem (Theorem \ref{thm:R.cycle}).

\begin{proposition}\label{prop:digraph}
Let $D=(V,A)$ be any directed graph. If each $v\in V$ satisfies $(\mathrm{indeg}_D(v), \mathrm{outdeg}_D(v))$$\in \{(0,2), (1,1), (1,2), (2,1), (2,0)\}$, then 
$r=|V_2|-|V_0|+1$ holds, where $r:=|A|-|V|+1$ and $V_2$ (resp. $V_0$) denotes the set of vertices $v$ with $\mathrm{indeg}_D(v)=2$ (resp. $\mathrm{indeg}_D(v)=0$). If, in addition, $D$ is acyclic, then $r \leq |V_2|$ holds.
\end{proposition}

\begin{proof}
Let $V_1:=V \setminus (V_0 \cup V_2)$. 
Since the sum of in-degrees of all vertices must equal $|A|$ , we have $|A| = |V_1| + 2|V_2|$. 
Hence,  $r = |A| - |V| +1 = |V_1|+2|V_2| - (|V_0| + |V_1| + |V_2|) + 1 = |V_2|-|V_0|+1$. 

We claim $|V_0| \geq 1$ holds if $D$ is acyclic. To verify this, we will prove that any finite directed acyclic graph $G$ has at least one vertex of in-degree zero. Let $P$ be a longest vertex-disjoint directed path in $G$, with $s$ and $t$ as the start and end vertices of $P$, respectively. If $\mathrm{indeg}_G(s)>0$, there exists a vertex $s^\prime \in V(G)$ with $(s^\prime, s)  \in A(G)$. Since $G$ is acyclic, $s^\prime$ is not a vertex of $P$. Then, by adding $(s^\prime,s)$ to $P$, one can obtain a vertex-disjoint directed path in $G$ that is longer than $P$, but this contradicts the maximality of $P$. Thus, $\mathrm{indeg}_G(s)=0$ holds, which proves the above claim. Hence, we obtain $r \leq |V_2|$.
\end{proof}

\begin{theorem}\label{thm:R.cycle}
	Let $(N, e_\rho, \delta_N^-)$ be an instance of Problem \ref{prob:constrained orientation} where $N$ is binary and let $V_R$ denote the set of reticulations specified by $\delta_N^- (v)$. If there exists an orientation $\vec{N}$ of $N$ satisfying the constraint $(e_{\rho}, \delta_N^-)$, then for any cycle basis $\mathcal{S}$ of $N$, there exists a bijection $\phi: \mathcal{S} \to V_R$ with the property that $\phi(C) \in V(C)$ holds for each $C \in \mathcal{S}$.   
\end{theorem}

\begin{proof}
Let $\mathcal{S}$ be any cycle basis of $N$ and let $B$ be the bipartite graph defined by $V(B):=\mathcal{S} \sqcup V_R$ and $E(B):=\{(C, v) \in  \mathcal{S} \times V_R \mid  v\in V(C)\}$. 
Here, $|\mathcal{S}| = |V_R|$ holds (note that we could have $|\mathcal{S}| \geq |V_R|$ if $N$ were allowed to have a reticulation with a large in-degree but here $N$ is binary).  
We also see that no vertex of $B$ has degree zero for the following reasons: since $\vec{N}$ is a directed acyclic graph, Lemma \ref{lem:R.position} implies that for each cycle $C \in \mathcal{S}$, there exists at least one reticulation $v \in V_R$ with $v \in V(C)$; conversely, for each reticulation $v\in V_R$, there exists at least one cycle $C \in \mathcal{S}$ with $v\in V(C)$. The proof will be completed if we can show that there is a perfect matching in $B$. We will prove that $B$ satisfies the marriage condition, i.e. for any subset  $\mathcal{S}'$ of $\mathcal{S}$, $|\mathcal{S}^\prime| \leq |V_R^\prime|$ holds, where $V_R^\prime$ is the neighbourhood of $\mathcal{S}^\prime$ in $B$. To prove this, without loss of generality, we may assume that the union of all cycles in $\mathcal{S}^\prime$, denoted by $N^\prime$,  is a connected subgraph of $N$. 

Recalling that $\vec{N}$ is a unique acyclic orientation for $(N, e_\rho, V_R)$, one can convert $N^\prime$ into  a directed  graph $D$ by assigning the same direction to each edge of $N^\prime$ as in $\vec{N}$. Since $D$ is a subgraph of $\vec{N}$, $D$ is also acyclic. By construction, each vertex $u$ of $D$ satisfies 
$(\mathrm{indeg}_D(u), \mathrm{outdeg}_D(u))$$\in \{(0,2), (1,1), (1,2), (2,1), (2,0)\}$. 
If we write $V_2$ for the set $\{u\in V(D) \mid \mathrm{indeg}_D(u) =2\}$, then Proposition \ref{prop:digraph} yields $ |A(D)|-|V(D)|+1 \leq |V_2|$. The left hand side equals $|E(N^\prime)|-|V(N^\prime)|+1 = |\mathcal{S}^\prime| $. Since the in-degree of a vertex in a subgraph $D$ never exceeds its in-degree in the original graph $\vec{N}$, we have $V_2 \subseteq V_R^\prime$. Thus, $|V_2| \leq |V_R^\prime|$ holds.  
Hence, by Hall's marriage theorem, $B$ has a perfect matching. Any perfect matching in $B$ induces a bijection $\phi: \mathcal{S} \to V_R$ satisfying $\phi(C) \in V(C)$ for each $C \in \mathcal{S}$. 
\end{proof}

Theorem \ref{thm:R.cycle} provides a necessary condition for the feasibility of a reticulation placement $V_R$ in Problem \ref{prob:constrained orientation} (binary version), where feasibility means that  $(N, e_\rho, V_R)$ admits an acyclic orientation for some $e_\rho$. 
This implies that we need not consider all $\binom{|V|-|X|}{r}$ reticulation placements in solving Problem \ref{prob:C-orientation}, regardless of the class $C$ we are interested in. Therefore, we may use an arbitrary cycle basis $\mathcal{S}$ of $N$  and choose exactly one reticulation vertex from each of the $r$ basic cycles, thereby reducing the search space for feasible reticulation placements in both   Problems \ref{prob:C-orientation} and  \ref{prob:tree-child orientation}. 

In solving Problem \ref{prob:tree-child orientation}, we will also use the following results. 
From the forbidden structures of tree-child networks (Fig. \ref{fig:forbidden}), Lemma \ref{lem:R.distance} follows immediately, and this leads to Theorem \ref{thm:dist3.TC}.

\begin{lemma}
	    \label{lem:R.distance}
A directed phylogenetic network $N$ is tree-child if every two reticulations of $N$ are \mypink{at a distance of} at least $3$ in the underlying graph of $N$. 
\end{lemma}

\begin{theorem}
	\label{thm:dist3.TC}
If $(N, e_\rho, \delta_N^-)$ is an instance of Problem \ref{prob:constrained orientation} such that $d_N(u,v) \geq 3$ for any distinct $u, v\in \{v \in V(N) \mid \delta_N^-(v)=2\}$ and if there exists an orientation $\vec{N}$ of $N$ that satisfies the constraint $(e_\rho, \delta_N^-)$, then $\vec{N}$ is a tree-child network.
\end{theorem}

We note that Theorem \ref{thm:dist3.TC} provides a sufficient condition for  tree-child orientability, not a necessary condition. In fact, a tree-child network can contain a pair of reticulations whose distance is less than 3 (see Fig. \ref{fig:R.distance}).

\begin{figure}[h]
    \begin{center}
    \includegraphics[width=0.8\textwidth]{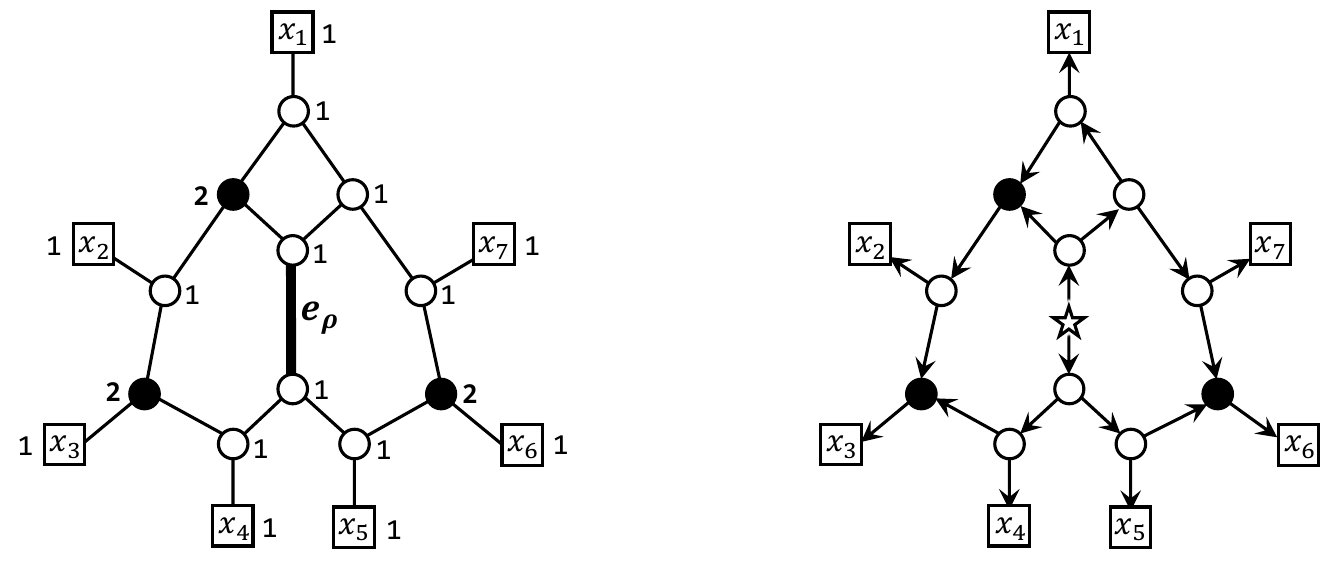}
    \caption{
        Left: An instance $(N, e_\rho, \delta_N^-)$ of Problem \ref{prob:constrained orientation} containing a pair of prescribed reticulations (i.e. vertices of desired in-degree $2$) at distance $2$. The squares and circles are the leaves and internal vertices, respectively. The root edge $e_\rho$ is highlighted in bold. The number next to each vertex $v$ indicates its desired in-degree $\delta_N^-(v)$, and the black vertices are the prescribed reticulations (i.e. those with desired in-degree $2$).
        Right: The orientation $\vec{N}$ for this instance $(N, e_\rho, \delta_N^-)$, which is a tree-child network. The root $\rho$ inserted into $e_\rho$ is shown as a unique star vertex.}
    \label{fig:R.distance}
    \end{center}
\end{figure}

\section{Proposed methods}\label{sec:proposed.methods}
Based on Theorem \ref{thm:R.cycle}, we propose an exact method (Algorithm  \ref{alg:proposed.BF}) for $\mathcal{C}$-\textsc{Orientation} (Problem \ref{prob:C-orientation}) and  a heuristic method (Algorithm \ref{alg:proposed.fast}) for \textsc{Tree-Child Orientation} (Problem \ref{prob:tree-child orientation}). While Theorem \ref{thm:R.cycle} holds for any cycle basis, our algorithms use a `minimal' one for computational efficiency. A cycle basis $\mathcal{S}$ of $N$ is called \emph{minimal} if the sum of the lengths of all cycles in $\mathcal{S}$ is not greater than that of any other cycle basis of $N$, i.e. minimising $\sum_{C\in \mathcal{S}}{|E(C)|}$. Intuitively, focusing on a minimal cycle basis allows us to use basic cycles that do not overlap too much. The problem of computing a minimal cycle basis has been extensively studied (e.g. \cite{stepanets1964basis, deo1982algorithms, horton1987polynomial}; see Chapter 7 of \cite{dePinaPhD} for a brief literature review) and polynomial time algorithms exist (e.g. \cite{dePinaPhD, amaldi2010efficient,horton1987polynomial}). By using an algorithm given in \cite{amaldi2010efficient}, a minimal cycle basis of $N=(V,E)$ can be computed in $O(|V||E|^2/ \log |V|)$.

For both Problems \ref{prob:C-orientation} and \ref{prob:tree-child orientation}, our proposed methods first compute a minimal cycle basis $\mathcal{S}=\{C_1,\dots, C_r\}$ of a given network $N$. Then, for Problem \ref{prob:C-orientation}, our exact method (Algorithm  \ref{alg:proposed.BF}) repeatedly picks exactly one reticulation vertex $v_i$ from each basic cycle $C_i$ to obtain a candidate set $V_R$ of $r$ reticulations, and solves Problem \ref{prob:constrained orientation} for all $(e_\rho, V_R)$ until it finds a $\mathcal{C}$-orientation of $N$. While this algorithm still requires exponential time, it is more efficient than the exponential time method in \cite{huber2024orienting} because the search space for appropriate $V_R$ is smaller than $\binom{O(|V|)}{r}$.

For Problem \ref{prob:tree-child orientation}, we first note that a straightforward way to reduce the search space for $V_R$ is to exclude any candidate sets containing adjacent reticulations, as such configurations are invalid in tree-child networks (recall Fig. \ref{fig:forbidden}, right panel). Beyond this basic constraint, we introduce a novel approach to further reduce the search space based on Theorem \ref{thm:dist3.TC}, which suggests that reticulations should be placed as far apart as possible. This insight leads to our heuristic method (Algorithm \ref{alg:proposed.fast}) that considers only those reticulation sets $V_R$ that maximise the sum of pairwise distances between reticulations. 
While this significant reduction in the search space makes the algorithm faster than Algorithm \ref{alg:proposed.BF}, it may fail to find an existing tree-child orientation; more precisely, while false positives cannot occur, a `NO' output from Algorithm \ref{alg:proposed.fast} should be interpreted as `Probably NO'.
However, Theorem \ref{thm:exact} ensures that Algorithm \ref{alg:proposed.fast} works correctly when $r$ is very small.

\begin{algorithm}
    \caption{Exact FPT Algorithm for $\mathcal{C}$-\textsc{Orientation}}
    \label{alg:proposed.BF}
    \begin{algorithmic}[1]
        \Require An undirected binary phylogenetic network $N = (V, E)$ on $X$
        \Ensure A \mypink{$\mathcal{C}$}-orientation $\vec{N}$ of $N$ if one exists, else `NO'
        \State Compute the number $r:=|E|-|V|+1$ of reticulations $\vec{N}$ must have
        \State Compute a minimal cycle basis $\mathcal{S}=\{C_1,\dots, C_r\}$ of $N$
         \State Compute $S := \{(v_1,\dots,v_r) \in V(C_1) \times \dots \times V(C_r)\}$
        \For{each reticulation placement $s=(v_1, \dots, v_r) \in S$}
            \For{each vertex $v\in V$}
                \State Define the desired in-degree $\delta_N^-(v)$ as $\delta_N^-{(v)}:=2$ if $v \in \{v_1, \dots, v_r\}$ and $\delta_N^-{(v)}:=1$ otherwise
            \EndFor
            \Repeat
                \State Pick any $e\in E$, set $e_\rho:=e$
                \State Run the linear time algorithm for Problem~\ref{prob:constrained orientation} in \cite{huber2024orienting} for the instance $(N, e_\rho, \delta_N^-)$
                \If{the algorithm finds the feasible orientation $\tilde{N}$ for $(N, e_\rho, \delta_N^-)$ and $\tilde{N}$ is in $\mathcal{C}$}
                    \State \Return $\tilde{N}$ as a $\mathcal{C}$-orientation $\vec{N}$ of $N$
                \EndIf
            \Until{no more edges are left in $E$}
        \EndFor
        \State \Return `NO'
    \end{algorithmic}
\end{algorithm}

\begin{algorithm}
    \caption{Heuristic Algorithm for \textsc{Tree-Child Orientation}}
    \label{alg:proposed.fast}
    \begin{algorithmic}[1]
        \Require An undirected binary phylogenetic network $N = (V, E)$ on $X$.
        \Ensure A tree-child orientation $\vec{N}$ of $N$ if found, else `NO'.
        \State Compute the number $r:=|E|-|V|+1$ of reticulations $\vec{N}$ must have
        \State Compute a minimal cycle basis $\mathcal{S}=\{C_1,\dots, C_r\}$ of $N$
        \State Compute $S := \{(v_1,\dots,v_r) \in V(C_1) \times \dots \times V(C_r) \mid d_N(v_i, v_j)\geq 2\}$
        \State Compute $S^\ast := \{s^\ast \in S \mid s^\ast = \argmax_{s\in S}{f(s):=\sum_{1\leq i<j\leq r}{d_N(v_i, v_j)}} \}$
        \For{each reticulation placement $s^\ast = (v_1^\ast, \dots, v_r^\ast) \in S^\ast$}
            \State Define the desired in-degree $\delta_N^-{(v)}$ of all $v\in V$ as $\delta_N^-{(v)}:=2$ if $v\in \{v_1^\ast,\dots, v_r^\ast\}$ and otherwise $\delta_N^-{(v)}:=1$
        \Repeat
    \State Pick any $e\in E$, set $e_\rho:=e$
    \State Run the linear time algorithm for Problem~\ref{prob:constrained orientation} in \cite{huber2024orienting} for the instance $(N, e_\rho, \delta_N^-)$
    \If{the algorithm finds the feasible orientation $\tilde{N}$ for $(N, e_\rho, \delta_N^-)$ and no vertex of $\tilde{N}$ has only reticulations as its children}
        \State \Return $\tilde{N}$  as a tree-child orientation $\vec{N}$ of $N$
    \EndIf
\Until{no more edges are left in $E$}
        \EndFor
\State \Return `NO'
    \end{algorithmic}
\end{algorithm}

\begin{theorem}\label{thm:exact}
If the input $N=(V,E)$ satisfies $r:=|E|-|V|+1 \leq 2$, then  Algorithm \ref{alg:proposed.fast} returns a correct solution to \textsc{Tree-Child Orientation} (Problem \ref{prob:tree-child orientation}).
\end{theorem}
\begin{proof}
\mypink{Any $N$ with $r\leq 2$ is acyclically orientable. When $r \leq 1$, it is clear that $N$ is tree-child orientable since any orientation of $N$ cannot contain any of the forbidden structures (Fig. \ref{fig:forbidden}). Therefore,} we may focus on the case of $r = 2$. When $r = 2$, $N$ contains exactly two cycles $C_1$ and $C_2$, each of which has at least $3$ edges. 
When $C_1$ and $C_2$ do not share an edge, $N$ is a level-$1$ network \mypink{(i.e., a network such that each biconnected component contains at most one cycle)}.
This implies that $N$ is planar and tree-child orientable (see Fig. \ref{fig:proof.cycles.disjoint}).  Algorithm \ref{alg:proposed.fast} computes a minimal cycle basis $\{C_1, C_2\}$ of $N$, which is unique in this case. Then, Algorithm \ref{alg:proposed.fast} selects a most distant pair $s^\ast=(v_1^\ast, v_2^\ast)$ in $V(C_1)\times V(C_2)$. Since $N$ is binary, $d_N(v_1^\ast, v_2^\ast) \geq 3$. 
We can see that $(N, e_\rho, \delta_N^-)$ has an orientation $\vec{N}$ for \mypink{a} root edge $e_\rho$ \mypink{in the path between $C_1$ and $C_2$, as} shown in Fig.~\ref{fig:proof.cycles.disjoint}.
By Theorem \ref{thm:dist3.TC},  $\vec{N}$ is tree-child. Hence, Algorithm \ref{alg:proposed.fast} can correctly find a tree-child orientation $\vec{N}$ of $N$.  


When $C_1$ and $C_2$ share an edge, they share exactly one \mypink{path $P$} because $r\geq 3$ otherwise. \mypink{Let $k$ be the length of $P$. 
When $k = 1$,} as Fig. \ref{fig:proof.cycles.not.disjoint} indicates, $N$ has a tree-child orientation $\vec{N}$ if and only if at least one of $C_1$ and $C_2$ has $4$ or more edges.  When each of $C_1$ and $C_2$ has exactly $3$ edges as in Fig. \ref{fig:proof.cycles.not.disjoint}(a), Algorithm \ref{alg:proposed.fast} correctly returns `NO'. When $|E(C_1)|=3$ and $|E(C_2)|=4$,  the algorithm selects a most distant pair $s^\ast=(v_1^\ast, v_2^\ast)$ in $V(C_1)\times V(C_2)$ as in Fig.~\ref{fig:proof.cycles.not.disjoint}(b). Although $d_N(v_1^\ast, v_2^\ast) = 2$ holds, the algorithm can insert the root $\rho$ into an appropriate edge $e_\rho \in E$, making their distance from $2$ to $3$. When $|E(C_1)|\geq 4$ and $|E(C_2)| \geq 4$ as in Fig. \ref{fig:proof.cycles.not.disjoint}(c), $d_N(v_1^\ast, v_2^\ast) \geq 3$. Hence, when a tree-child orientation $\vec{N}$ of $N$ exists, Algorithm \ref{alg:proposed.fast} correctly outputs it.

\mypink{Likewise, when $k = 2$, we can see that $N$ has a tree-child orientation $\vec{N}$ if and only if at least one of $C_1$ and $C_2$ has $5$ or more edges  \mypink{(see Fig.\ \ref{fig:proof.cycles.not.disjoint}(d)(e))} and that Algorithm \ref{alg:proposed.fast} outputs a correct solution $\vec{N}$ if it exists, similarly to the previous case.}

\mypink{
Suppose $k \geq 3$. Since $\{C_1, C_2\}$ is a minimal cycle basis of $N$, the lengths of the two paths other than $P$ connecting the two endpoints of $P$ are both $k$ or more. Then, as Algorithm \ref{alg:proposed.fast} selects a most distant reticulation pair $(v_1^\ast, v_2^\ast)$ in $V(C_1)\times V(C_2)$,  we have $d_N(v_1^\ast, v_2^\ast) \geq 3$. This implies that in any oriented network output by Algorithm \ref{alg:proposed.fast}, the two reticulations are at a distance $3$ or more. Hence, by Lemma \ref{lem:R.distance}, if a tree-child orientation $\vec{N}$ of $N$ exists, Algorithm \ref{alg:proposed.fast} correctly outputs it.}
\end{proof}
  
\begin{figure}[htbp]
    \begin{center}
    \includegraphics[width=0.5\textwidth]{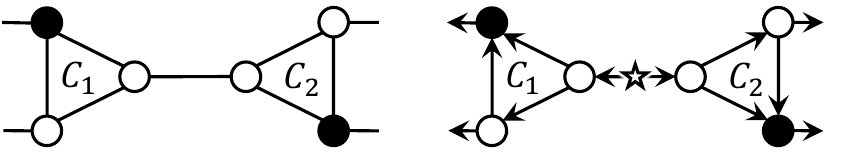}
    \caption{
    Proof of Theorem \ref{thm:exact} (the case when two cycles are edge-disjoint). The star is the root $\rho$. The black vertices are a pair of reticulations, $v_1^\ast$ and $v_2^\ast$, that maximises the sum of their distances. 
  \label{fig:proof.cycles.disjoint}
    }
      \end{center}
\end{figure}

\begin{figure}[htbp]
    \begin{center}
    \includegraphics[width=0.9\textwidth]{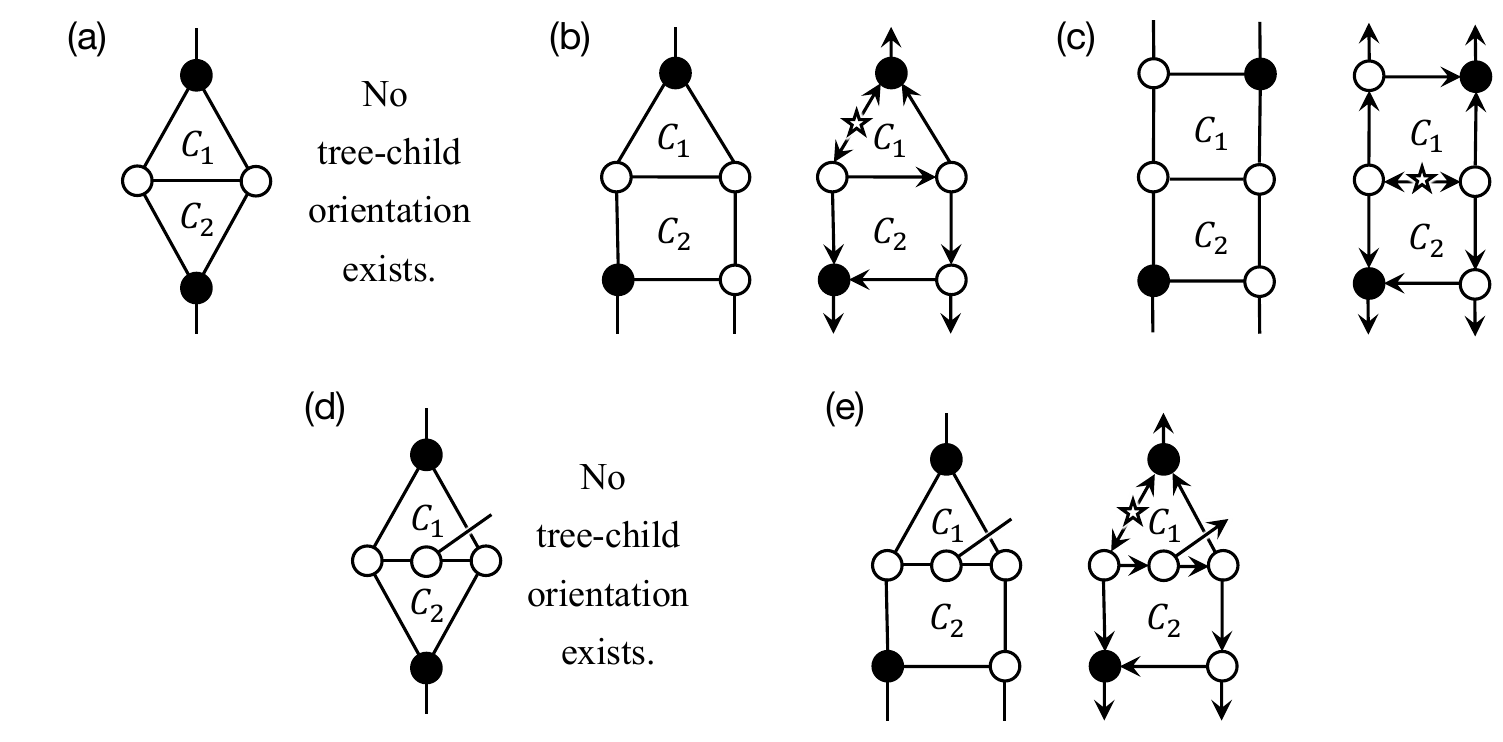}
    \caption{Proof of Theorem \ref{thm:exact} (\mypink{illustrating the cases where} \mypink{the} two cycles \mypink{share a path of length $k\in \{1,2\}$}). The star is the root $\rho$. The black vertices are a pair of reticulations, $v_1^\ast$ and $v_2^\ast$, that maximises the sum of their distances in each case. \mypink{For $k=1$:} (a) When both $C_1$ and $C_2$ are $3$-cycles, $N$ must be a NO instance; (b) \mypink{and} (c) \mypink{are} YES instances. \mypink{For $k=2$: (d) When both $C_1$ and $C_2$ are $4$-cycles, $N$ must be a NO instance; (e) is a YES instance.}
\label{fig:proof.cycles.not.disjoint}
        }
    \end{center}
\end{figure}

Theorem \ref{thm:complexity} states that Algorithm \ref{alg:proposed.BF} is FPT both in the reticulation number $r$ and the size $c$ of \mypink{a} longest basic \mypink{cycle} in $\mathcal{S}$ used in the computation.

\begin{theorem}\label{thm:complexity}
	Algorithm \ref{alg:proposed.BF} solves \textsc{Tree-Child Orientation} (Problem \ref{prob:tree-child orientation}) in $O(c^r \cdot |V||E|) =O(c^r \cdot |V|^2)$ time, 
	where $r$ is the reticulation number of $N=(V,E)$ and $c$ is the size of \mypink{a} largest \mypink{cycle} in a minimal cycle basis $\mathcal{S}$ chosen at Line 2.
\end{theorem}

\begin{proof}
Recalling a minimal cycle basis $\mathcal{S}$ of $N$  can be computed in  $O(|V||E|^2/ \log |V|)$  time by the algorithm in \cite{amaldi2010efficient}, we know that Line 2 takes $O(|V||E|^2/ \log |V|)$ time. 
 Line 5--7 takes $O(|V|)$ time for each $s\in S$.  Also, since one can check in $O(|V|)$ time whether or not $\tilde{N}$ is in the class $\mathcal{C}$ of tree-child networks, Line 8--14 takes $O(|E|\times (|V|+|V|))=O(|V||E|)$ time for each $s\in S$. Line 5--14 needs to be repeated $|S|=O(|V(C_1)|\times \dots \times |V(C_r)|) = O(c^r)$ times. Therefore, Line 4--15 can be done in $O(c^r \cdot |V||E|)$ time. 
Since $N$ is binary, $O(|V|)=O(|E|)$ holds. Thus, Algorithm \ref{alg:proposed.BF} runs in  $O(c^r\cdot |V|^2)$ time.
\end{proof}

Theorem \ref{thm:complexity} implies that the unparameterised worst-case complexity of Algorithm~\ref{alg:proposed.BF} is $O(|V|^{r+2})$. This complexity is equivalent to that of the exponential time method (Algorithm~2 in \cite{huber2024orienting}) that performs $O(|V||E|)$ time calculations for all $\binom{|V|}{r}=O(|V|^r)$ reticulation placements. Although two FPT algorithms for a special case of $\mathcal{C}$-\textsc{Orientation} were proposed in \cite{huber2024orienting}, Algorithm \ref{alg:proposed.BF} differs in its parameterisation from them. Specifically, our Algorithm \ref{alg:proposed.BF} is parameterised by $r$ and $c$, while the FPT algorithms in \cite{huber2024orienting} are parameterised by $r$ and by the level of $N$, respectively.
%
%

Theorem \ref{thm:complexity} also shows that the size of search space depends on the choice of $\mathcal{S}$ at Line 2 of Algorithm \ref{alg:proposed.BF} although $c$ never exceeds the size of longest cycles in $N$ (the same applies to Line~2 of Algorithm \ref{alg:proposed.fast}). To illustrate this, consider two minimal cycle bases $\mathcal{S}=\{C_1, C_2, C_3\}$ with $|V(C_1)|=|V(C_2)|=|V(C_3)|=4$ and $\mathcal{S}^\prime=\{C_1^\prime, C_2^\prime, C_3^\prime\}$ where $|V(C_1^\prime)|=3$, $|V(C_2^\prime)|=4$ and $|V(C_3^\prime)|=5$ (note that they have the same total length $4+4+4=3+4+5$). When the former $\mathcal{S}$ is selected, the number of elements of $S$ at Line 3 of either algorithm is $|V(C_1)|\times |V(C_2)|\times |V(C_3)|=4^3$, whereas $|S|=|V(C_1^\prime)|\times |V(C_2^\prime)|\times |V(C_3^\prime)|=3\times 4\times 5$ for the latter $\mathcal{S}^\prime$.


\section{Experiments}\label{sec:experiments}
We implemented our two proposed methods (\mypink{the \textsc{Tree-Child Orientation} version of} Algorithm \ref{alg:proposed.BF}, and Algorithm \ref{alg:proposed.fast}) and the existing exponential time algorithm described by Huber {\textit{et al.}}\ (Algorithm 2 in \cite{huber2024orienting}) using Python 3.11.6. 
In the implementation of Algorithms \ref{alg:proposed.BF} and \ref{alg:proposed.fast}, we used the \texttt{minimum\_cycle\_basis} function from the Python package \texttt{networkx} to compute a minimal cycle basis. We note that the algorithm implemented in \texttt{networkx} is not the $O(|V||E|^2/ \log |V|)$ algorithm given in \cite{amaldi2010efficient} but the $O(|E|^3+|E||V|^2\log{|V|})$ algorithm given in Section 7.2 of \cite{dePinaPhD}.
Undirected binary phylogenetic networks were generated as test data using a method described in Appendix.  
The source code, test data, and the program used to generate the data are available at \url{https://github.com/hayamizu-lab/tree-child-orienter}. The full details of the results can be found at \url{https://github.com/hayamizu-lab/tree-child-orienter/tree/main/results}.


\subsection{Experiment 1: Execution time}
\experimentlabel{expt:time}
In Experiment \ref{expt:time}, we compared the execution times of the following methods: the existing exponential time algorithm (Algorithm 2 in \cite{huber2024orienting}),  Algorithm \ref{alg:proposed.BF}, and Algorithm~\ref{alg:proposed.fast} on 20 hand-picked tree-child orientable networks with 10 leaves. The 20 networks consisted of five samples each for reticulation numbers $r=2,3,4$ and $5$. Due to the time-consuming nature of the existing method, conducting experiments with a larger number of samples was infeasible. The experiment was performed on a MacBook Air (CPU: Intel Core i5, 1.6GHz, 8GB memory). We note that the hardware specification in Experiment \ref{expt:time} differs from Experiments \ref{expt:small_graph} and \ref{expt:large_graph}, but this does not undermine the validity of this study as we do not compare results across experiments.

The results are summarised in Fig. \ref{fig:barchart}. 
Although both Algorithm 2 in \cite{huber2024orienting} and Algorithm~\ref{alg:proposed.BF} require exponential time in general, Algorithm \ref{alg:proposed.BF} is expected to be faster in practice due to its smaller search space. Indeed, Algorithm \ref{alg:proposed.BF} was significantly faster than the existing exponential time method. We also confirmed that Algorithm \ref{alg:proposed.fast} was faster than Algorithm  \ref{alg:proposed.BF}.

\begin{figure}[htbp]
    \begin{center}
\includegraphics[width=.85\textwidth]{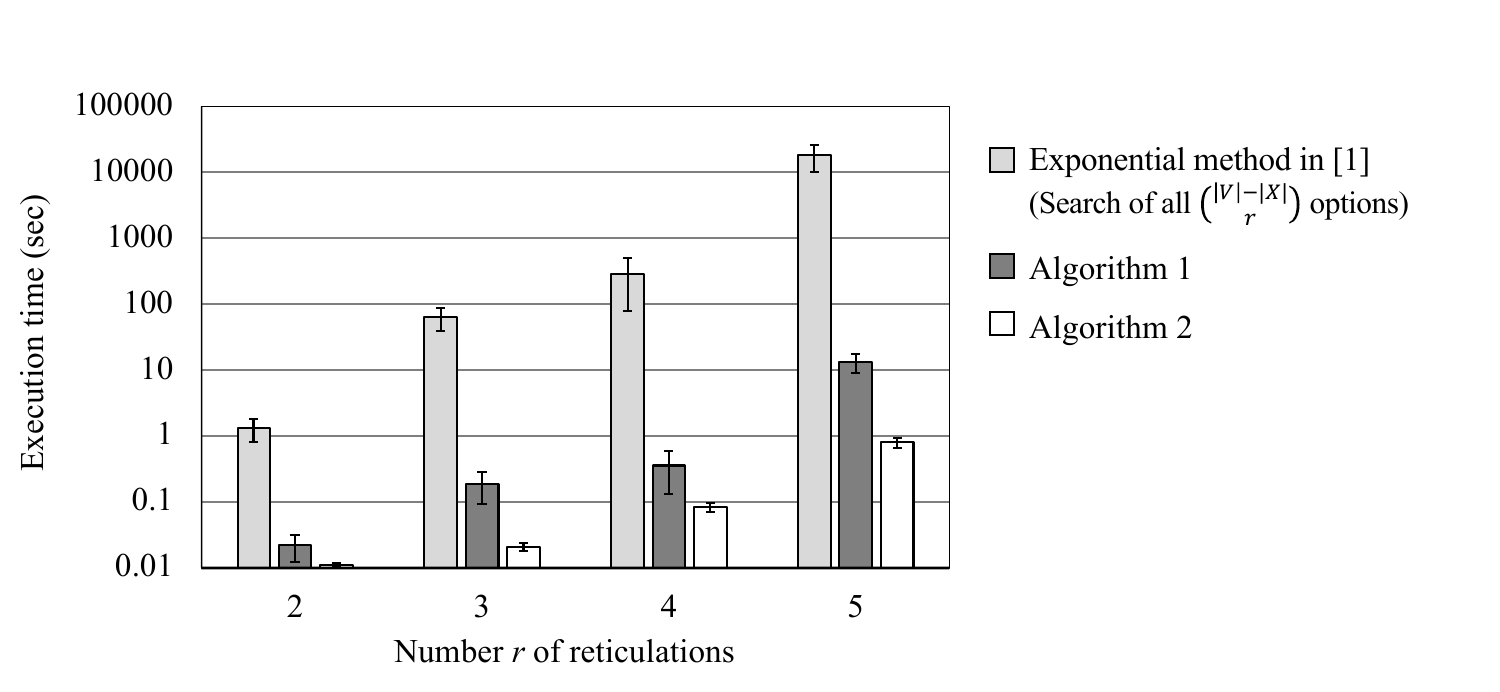}
        \caption{Results of Experiment \ref{expt:time}. Comparison of the execution time of the three methods for tree-child orientable networks with $10$ leaves. Our exponential time algorithm (Algorithm \ref{alg:proposed.BF}) has been shown to be significantly faster than the existing exponential time algorithm described in \cite{huber2024orienting}.
    \label{fig:barchart}
    }
    \end{center}
\end{figure}

\subsection{Experiment 2: Accuracy for small graphs}
\experimentlabel{expt:small_graph}
In Experiment \ref{expt:small_graph}, we evaluated the accuracy and execution time of Algorithms \ref{alg:proposed.BF} and \ref{alg:proposed.fast} on the 289 networks with 10 leaves and 1 to 5 reticulations in Table \ref{table:materials}. The experiment was performed on a MacBook Pro (CPU: Apple M2 Pro, clock speed 3.49GHz, memory 16GB). The tree-child orientability of the sample graphs was determined using the existing exponential time algorithm (Algorithm 2 in \cite{huber2024orienting}). Out of 289 instances, 268 were tree-child orientable (YES instances) and 21 were not (NO instances).

The accuracy and execution time of the two methods are summarised in Table \ref{table:accuracy}.
The correctness of Algorithm \ref{alg:proposed.BF} was empirically confirmed, while being much faster than the existing method in \cite{huber2024orienting}. In accordance with Theorem \ref{thm:exact}, Algorithm \ref{alg:proposed.fast} returned a correct solution whenever $r\leq 2$, but in practice, it still worked correctly for most cases with  $r \leq 4$, with a much shorter running time than Algorithm \ref{alg:proposed.BF} for both YES and NO instances. Algorithm \ref{alg:proposed.fast} often failed to find a tree-child orientation for YES instances with $r = 5$. 

\subsection{Experiment 3: Accuracy for large graphs}
\experimentlabel{expt:large_graph}
In Experiment \ref{expt:large_graph}, we evaluated the performance of Algorithms \ref{alg:proposed.BF} and \ref{alg:proposed.fast} on the 471 larger sample networks with 20 leaves and 1 to 9 reticulations in Table \ref{table:materials}. 
The experiment was performed on the same MacBook Pro used in Experiment \ref{expt:small_graph}. Since Algorithm 2 in \cite{huber2024orienting} became \mypink{intractable} for most cases with $r \geq 6$, we compared Algorithms \ref{alg:proposed.BF} and \ref{alg:proposed.fast} based on the number of tree-child orientations found and the time taken to find one. 

The results are summarised in Table \ref{table:TCfound}.
As reticulation number $r$ increased, the ability of Algorithm \ref{alg:proposed.fast} to find tree-child orientations decreased monotonically. By contrast, Algorithm \ref{alg:proposed.BF} was still able to find a tree-child orientation for many instances with large $r$. It generally takes a long time, but sometimes it can find a tree-child orientation in a practical time.

\begin{table}[htbp]
    \caption{
        Results of Experiment \ref{expt:small_graph}. The input networks have $10$ leaves.  \#YES (resp.\ \#NO) is the number of tree-child orientable (resp.\ non-tree-child orientable) instances among the generated networks with each reticulation number $r$. \#YES and \#NO have been verified using the existing exponential time algorithm in \cite{huber2024orienting}. The execution time here is the time taken to find a tree-child orientation or to output `NO' or `Probably NO'.  
         \label{table:accuracy}
        }
    \begin{tabular}{crcrrrcrrr}
        \hline
         &  & \multicolumn{4}{c}{Algorithm \ref{alg:proposed.BF}} & \multicolumn{4}{c}{Algorithm \ref{alg:proposed.fast}} \\ \hline
        \multirow{2}{*}{$r$} & \#YES & \multirow{2}{*}{Accuracy} & \multicolumn{3}{c}{Execution Time (sec)} & \multirow{2}{*}{Accuracy} & \multicolumn{3}{c}{Execution Time (sec)} \\
         & \#NO &  & Mean & Min & Max &  & Mean & Min & Max \\ \hline
         \multirow{2}{*}{1} & 170 & \begin{tabular}[c]{@{}c@{}}170/170\\ (100\%)\end{tabular} & 0.002 & 0.001 & 0.004 & \begin{tabular}[c]{@{}c@{}}170/170\\ (100\%)\end{tabular} & 0.002 & 0.001 & 0.004 \\
         & 0 & N/A & N/A & N/A & N/A & N/A & N/A & N/A & N/A \\ \hline
        \multirow{2}{*}{2} & 52 & \begin{tabular}[c]{@{}c@{}}52/52\\ (100\%)\end{tabular} & 0.011 & 0.002 & 0.063 & \begin{tabular}[c]{@{}c@{}}52/52\\ (100\%)\end{tabular} & 0.004 & 0.002 & 0.018 \\
         & 5 & \begin{tabular}[c]{@{}c@{}}5/5\\ (100\%)\end{tabular} & 0.062 & 0.061 & 0.064 & \begin{tabular}[c]{@{}c@{}}5/5\\ (100\%)\end{tabular} & 0.013 & 0.007 & 0.036 \\ \hline
        \multirow{2}{*}{3} & 24 & \begin{tabular}[c]{@{}c@{}}24/24\\ (100\%)\end{tabular} & 0.081 & 0.004 & 0.497 & \begin{tabular}[c]{@{}c@{}}24/24\\ (100\%)\end{tabular} & 0.007 & 0.005 & 0.012 \\
         & 7 & \begin{tabular}[c]{@{}c@{}}7/7\\ (100\%)\end{tabular} & 0.539 & 0.232 & 1.173 & \begin{tabular}[c]{@{}c@{}}7/7\\ (100\%)\end{tabular} & 0.021 & 0.011 & 0.047 \\ \hline
        \multirow{2}{*}{4} & 17 & \begin{tabular}[c]{@{}c@{}}17/17\\ (100\%)\end{tabular} & 0.741 & 0.008 & 3.735 & \begin{tabular}[c]{@{}c@{}}16/17\\ (94\%)\end{tabular} & 0.026 & 0.013 & 0.056 \\
         & 6 & \begin{tabular}[c]{@{}c@{}}6/6\\ (100\%)\end{tabular} & 3.918 & 2.428 & 7.650 & \begin{tabular}[c]{@{}c@{}}6/6\\ (100\%)\end{tabular} & 0.048 & 0.020 & 0.112 \\ \hline
        \multirow{2}{*}{5} & 4 & \begin{tabular}[c]{@{}c@{}}4/4\\ (100\%)\end{tabular} & 14.382 & 1.481 & 48.980 & \begin{tabular}[c]{@{}c@{}}1/4\\ (25\%)\end{tabular} & 0.153 & 0.133 & 0.173 \\
         & 4 & \begin{tabular}[c]{@{}c@{}}4/4\\ (100\%)\end{tabular} & 23.340 & 13.892 & 39.959 & \begin{tabular}[c]{@{}c@{}}4/4\\ (100\%)\end{tabular} & 0.115 & 0.088 & 0.156 \\ \hline
        \end{tabular}
    \end{table}


\begin{table}[htbp]
\centering
\caption{
        Results of Experiment \ref{expt:large_graph}. The input networks have $20$ leaves. \#Graph is the number of generated networks with each reticulation number $r$. \#YES is the number of YES instances among those networks, which is equal to the number of tree-child orientations found by Algorithm \ref{alg:proposed.BF}.  \#TC is the number of tree-child orientations found by Algorithm \ref{alg:proposed.fast}. The execution time here is the time taken to find a tree-child orientation.  
        \label{table:TCfound}
}
{\footnotesize
\begin{tabular}{crcrrrcrrr}
\hline
             &  & \multicolumn{4}{c}{Algorithm \ref{alg:proposed.BF}} & \multicolumn{4}{c}{Algorithm \ref{alg:proposed.fast}} \\ \hline
            \multirow{2}{*}{$r$} & \multirow{2}{*}{\#Graph} & \multirow{2}{*}{\#YES} & \multicolumn{3}{c}{Execution time (sec)} & \multirow{2}{*}{\#TC} & \multicolumn{3}{c}{Execution Time (sec)} \\
             &  &  & Mean & Min & Max &  & Mean & Min & Max \\ \hline
            1 & 163 & 163 & 0.005 & 0.004 & 0.025 & 163 & 0.005 & 0.003 & 0.025 \\
            2 & 108 & 107 & 0.032 & 0.006 & 0.152 & 107 & 0.009 & 0.006 & 0.064 \\
            3 & 77 & 75 & 0.323 & 0.010 & 2.312 & 74 & 0.019 & 0.010 & 0.081 \\
            4 & 49 & 44 & 2.352 & 0.014 & 13.404 & 40 & 0.067 & 0.031 & 0.122 \\
            5 & 23 & 18 & 11.485 & 0.031 & 52.339 & 15 & 0.414 & 0.125 & 1.452 \\
            6 & 20 & 13 & 226.701 & 0.026 & 966.848 & 5 & 3.360 & 1.607 & 5.327 \\
            7 & 16 & 12 & 2732.563 & 1.190 & 23761.347 & 4 & 31.745 & 11.619 & 61.007 \\
            8 & 10 & 7 & 23155.714 & 1115.183 & 67313.690 & 1 & 942.973 & 942.973 & 942.973 \\
            9 & 5 & 3 & 229788.111 & 20176.763 & 623017.346 & 0 & N/A & N/A & N/A \\ \hline
\end{tabular}
}
\end{table}

\section{Discussion}\label{sec:discussion}
\mypink{As we have seen in Section~\ref{sec:experiments}, the heuristic method (Algorithm~\ref{alg:proposed.fast}) is significantly faster than the exact FPT method (Algorithm~\ref{alg:proposed.BF}), but it suffers from an increasing rate of false negatives as the reticulation number $r$ grows. We therefore examine the theoretical reasons behind its limitations and explore its potential biological applications.}

\subsection{Theoretical limitations of Algorithm \ref{alg:proposed.fast}}
Algorithm \ref{alg:proposed.fast} is very fast for both YES and NO instances, but interestingly, it becomes inaccurate as the reticulation number $r$ increases. It would be useful to analyse the possible causes of such failures. 

The first remark is that  the algorithm does not necessarily find a reticulation placement that maximises the sum of the pairwise distances, because it searches for an optimal placement for a fixed minimal cycle basis $\mathcal{S}$. More precisely, the choice of $\mathcal{S}$ can affect the maximum value  $f(s^\ast)$  of the objective function. For example, when the algorithm selects the minimal cycle basis $\mathcal{S}=\{C_1,\dots, C_8\}$ as shown on the left of Fig. \ref{fig:distsum_cyclebasis_difference}, an optimal reticulation placement $s^\ast$ attains $f(s^\ast)=130$. 
On the other hand, when the cycle $C_8$ is replaced as shown on the right, then $f(s^\ast)=129$. 
This observation suggests that if our goal is to maximise the sum of the distances between reticulations, then we need to carefully select a minimal cycle basis $\mathcal{S}$. 


\begin{figure}[htbp]
    \begin{center}
    \includegraphics[width=1.0\textwidth]{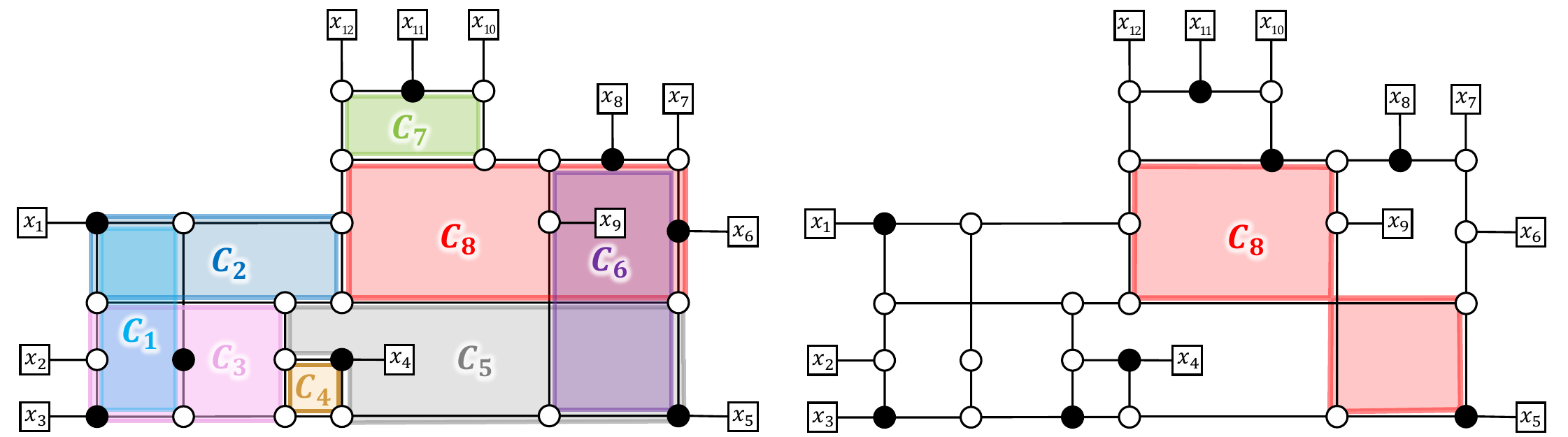}
    \caption{Discussion of the objective function $f$ of Algorithm \ref{alg:proposed.fast}. This network $N$ has multiple minimal cycle bases $\mathcal{S}$ with different values of $f(s^\ast)$. Different basic cycles are highlighted in different colours, and the black vertices indicate an optimal reticulation placement for each minimal cycle basis.  The reticulation set in the left figure achieves $f(s^*)=130$.  On the other hand, the placement in the right figure, where $C_8$ is replaced and the remaining cycles are the same, only attains $f(s^*)=129$. 
    }
    \label{fig:distsum_cyclebasis_difference}
    \end{center}
\end{figure}

However, more importantly, we note that maximising the sum of the pairwise distances of the reticulations is not always advantageous for finding a tree-child orientation. 
For example, the graph on the left of Fig. \ref{fig:failed.case} has a tree-child orientation if the four reticulations are placed as shown on the right. 
However, by maximising the sum of the distances of the reticulations, Algorithm \ref{alg:proposed.fast} has to select the reticulation placement with $f(s^\ast)=20$ as on the left of Fig. \ref{fig:failed.case}. Then, the algorithm will end up with returning `NO' because there is no tree-child orientation for this reticulation placement, regardless of the choice of root edge $e_\rho$. 
The placement on the right has $f(s)=19$, which is not maximum but does allow for a tree-child orientation. 

\begin{figure}[htbp]
    \begin{center}
    \includegraphics[width=0.7\textwidth]{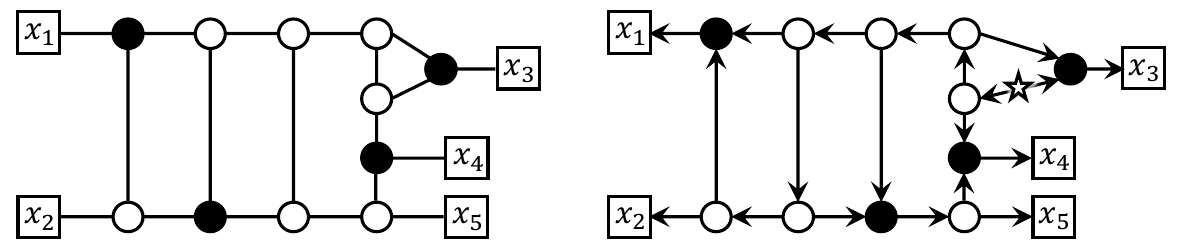}
    \caption{
        An example of $N$ that has a tree-child orientation but for which Algorithm \ref{alg:proposed.fast} fails. For any $s^\ast \in S^\ast$ maximising $f$ with $f(s^\ast)=20$,  $(N, e_\rho, \delta_N^-)$ is not tree-child orientable, regardless of the choice of root edge $e_\rho$. However, choosing $s\in S\setminus S^\ast$ as shown on the right, with $f(s)=19$, yields a tree-child orientable $(N, e_\rho, \delta_N^-)$.
        }
        \label{fig:failed.case}
    \end{center}
\end{figure}

\subsection{Biological application of Algorithm \ref{alg:proposed.fast}}\label{sec:bio}
While the biological implications of $\mathcal{C}$-\textsc{Orientation} methods require further investigation, we present here a case study applying our heuristic to a network in the biological literature. We emphasise that our focus is not on validating the biological accuracy, but rather on identifying challenges and areas for potential improvement.

Fig. \ref{fig:realdata}(a) shows a part of the biologically estimated tree-child network representing the evolutionary history of eight wheat relatives \cite{wheat}. Note that the original network proposed in \cite{wheat} is more complex, and we have ignored the uncertain reticulation arcs and selected some of the most likely reticulation arcs to create a tree-child network. Given the unrooted and undirected form of this network, Algorithm \ref{alg:proposed.fast} successfully determined it as tree-child orientable. However, the algorithm could not recover the same network as in Fig. \ref{fig:realdata}(a) and instead produced an alternative tree-child orientation with a different root position, such as the one shown in Fig. \ref{fig:realdata}(b).

It is true that the networks differ in their deep ancestral reticulation histories due to the alternative root positions, but they agree with more recent important reticulation events explained in \cite{wheat}, namely the introgression of the \textit{Sitopsis} ancestor by the \textit{Ae. speltoides} ancestor and the complex gene flows from \textit{Ae. tauschii} to \textit{Ae. caudata}. Notably, despite the difference in the root position, most arc orientations remain consistent with the original network. This observation underscores the importance of allowing root flexibility to find desired orientations that might otherwise be overlooked.

While we have seen the similarity between the networks, the result also shows a significant limitation of the heuristic. Indeed, even for the root placement shown in Fig. \ref{fig:realdata}(b), there exists a more accurate tree-child orientation that preserves the original reticulation structure, which can be obtained by reversing the red arc. The failure of Algorithm \ref{alg:proposed.fast} to identify the better solution stems from its current objective function, which prioritises maximising the sum of distances between reticulations. This observation further highlights the challenge of developing biologically meaningful orientation heuristics.

\begin{figure}[htbp]
\begin{center}
\includegraphics[width=1.0\textwidth]{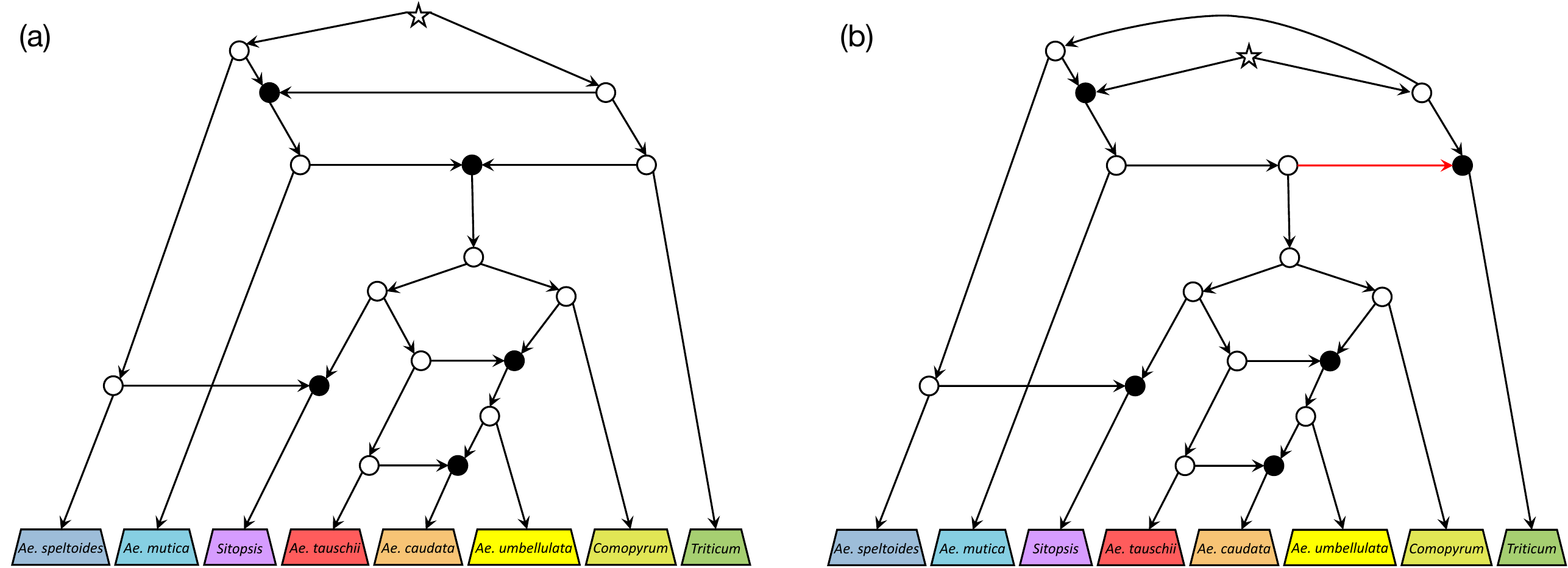}
\caption{Application of Algorithm \ref{alg:proposed.fast} to a biologically estimated network. (a) A tree-child network (created by a slight modification of Fig. 5 in \cite{wheat}) representing an evolutionary scenario for wheat relatives (from left to right: \textit{Ae. speltoides}, \textit{Ae. mutica}, \textit{Sitopsis}, \textit{Ae. tauschii}, \textit{Ae. caudata}, \textit{Ae. umbellulata}, \textit{Comopyrum} and \textit{Triticum}). (b) One of the possible outputs of Algorithm \ref{alg:proposed.fast} when applied to the unrooted and undirected version of network (a). The networks differ in root edge selection and the orientation of one edge (highlighted in red). Star vertices indicate roots, and black vertices denote nodes with in-degree two.}
\label{fig:realdata}
\end{center}
\end{figure}

\section{Conclusion and future work}\label{sec:conclusion.future.work}
The \mypink{$\mathcal{C}$}-\textsc{Orientation} problem, which asks whether a given undirected binary phylogenetic network can be oriented to a directed phylogenetic network of a desired class \mypink{$\mathcal{C}$}, is an important computational problem in phylogenetics. 
The complexity of this problem remains unknown for many network classes \mypink{$\mathcal{C}$}, including the class of binary tree-child networks \cite{bulteau2023turning, docker2025}. 
A simple exponential time algorithm for $\mathcal{C}$-\textsc{Orientation} was provided in \cite{huber2024orienting}. FPT algorithms for a special case of the problem were also proposed in \cite{huber2024orienting}, but their practical application is limited due to the intricate nature of the procedures and the challenges in implementation, despite the constraints imposed on $\mathcal{C}$. Additionally, no study has explored heuristic approaches to solve $\mathcal{C}$-\textsc{Orientation} in practice, even for a particular class $\mathcal{C}$ such as tree-child networks.  

In this paper, we have proposed a simple, easy to implement, practical exact FPT algorithm  (Algorithm~\ref{alg:proposed.BF}) for $\mathcal{C}$-\textsc{Orientation} and a  heuristic algorithm for \textsc{Tree-Child Orientation} (Algorithm \ref{alg:proposed.fast}) based on Theorem \ref{thm:R.cycle}. They improve the search space of the existing simple exponential  time algorithm by using a cycle basis to reduce the number of possible reticulation placements. Our experiments showed that Algorithm~\ref{alg:proposed.BF} is significantly faster in practice than a state-of-the-art exponential time  algorithm in \cite{huber2024orienting}. Algorithm \ref{alg:proposed.fast} is even faster, with a trade-off between the accuracy and the  reticulation number. Further research using larger and more diverse datasets could provide more insight into the strengths and limitations of the proposed methods. Their usefulness and effectiveness should also be tested in real-world data analysis. As tree-child networks can describe evolutionary histories that do not involve frequent reticulation events, one can find various tree-child networks in the literature other than the one we used in Section \ref{sec:bio}, such as a hybridisation network of bread wheat (Fig. 3 in \cite{TC.wheat2014}),  a hybridisation network of mosquitoes (Fig. 15 in \cite{TC.mosquito}) and an admixture network of diverse populations (Fig. 3 in \cite{TC.admixture}).

Although we have used \textsc{Tree-Child Orientation} as a case study for performance evaluation, Algorithm~\ref{alg:proposed.BF} is applicable to orientation problems for \mypink{any} classes $\mathcal{C}$ other than tree-child networks \mypink{as long as the membership test is tractable}. For example, orientation for stack-free networks \cite{semple2018phylogenetic} or tree-based networks \cite{10.1093/sysbio/syv037} is expected to be a good application, because one can quickly decide whether a given network belongs to such a class. 
Speeding up Algorithm \ref{alg:proposed.BF} and extending it to non-binary networks are topics for future research. It would be also interesting to consider alternative FPT algorithms that do not use a cycle basis.


Improving the accuracy of Algorithm~\ref{alg:proposed.fast} is also an interesting direction for future research. As discussed in Section \ref{sec:discussion}, there are clearly rooms for improvement in the current objective function. Introducing a more suitable objective function may lead to the development of faster and more useful tree-child orientation heuristics.

\begin{appendices}
\renewcommand\theHtable{AABB\arabic{table}}
\renewcommand\theHfigure{AABB\arabic{figure}}

\section{Algorithm for generating undirected binary phylogenetic networks}\label{secA1}
To create undirected binary phylogenetic networks on $X$ for the experiments, we used a simple method explained below (see  Table \ref{table:graph.generate} and Fig. \ref{fig:graph.generate} for illustrations). The code can be found at \url{https://github.com/hayamizu-lab/tree-child-orienter/tree/main/Appendix}. 
This method uses the idea of the coalescent model which is a popular approach for simulating phylogenetic trees. 

Given a set $X$ of $n$ present-day species, we trace back $n$ lineages from the present to the past. At each step, lineages can either split with probability $P_r$ or coalesce with probability $(1 - P_r)$. If two lineages coalesce, a tree vertex is created, and the set of extant taxa is updated. If a lineage splits, a reticulation is created, and the set of extant taxa is updated. This process continues until all lineages coalesce into a single vertex, the root. The resulting graph becomes a rooted directed binary phylogenetic network on $X$ after the vertices of in-degree and out-degree $1$ are suppressed and the necessary vertex and arc are added to resolve the nonbinary vertices. It can then be converted to an undirected phylogenetic network on $X$ after  suppressing the root and by ignoring all arc orientations and non-leaf vertex labels. 

By adjusting the value of $P_r$, we can generate phylogenetic networks with various reticulation numbers. When $P_r = 0$, no reticulations occur, and the generated networks are guaranteed to be phylogenetic trees. If $P_r$ has been set a larger value, the algorithm tends to produce graphs having more reticulations, as shown in Table \ref{table:materials}.  Thus, we generated undirected binary phylogenetic networks on $X$ with varying levels of complexity, and selected suitable ones in each computational experiment.

 For Experiments \ref{expt:small_graph} and \ref{expt:large_graph}, we generated $1200$ various undirected binary phylogenetic networks on $X$.
  The method requires the number $|X|$ of leaves and reticulation probability parameter $P_r$ to be specified. 
Table \ref{table:materials} summarises the breakdown of the generated graphs.

\begin{table}[htbp]
    \caption{An illustration of the graph generation method used in this study (see also Fig. \ref{fig:graph.generate}). 
        }
   \label{table:graph.generate}
    \begin{tabular}{@{}cccccc@{}}
                \toprule
                Step & Taxon set & Selected lineage(s) & Event & New taxa & New arcs \\ \midrule
                1 & $\lbrace1,2,3,4\rbrace$ & 3,4 & Coalesce & 5 & $(5,3),(5,4)$ \\
                2 & $\lbrace1,2,5\rbrace$ & 2 & Split & 6,7 & $(6,2),(7,2)$ \\
                3 & $\lbrace1,5,6,7\rbrace$ & 1,6 & Coalesce & 8 & $(8,1),(8,6)$ \\
                4 & $\lbrace5,7,8\rbrace$ & 7,8 & Coalesce & 9 & $(9,7),(9,8)$ \\
                5 & $\lbrace5,9\rbrace$ & 5 & Split & 10,11 & $(10,5),(11,5)$ \\
                6 & $\lbrace9,10,11\rbrace$ & 9,10 & Coalesce & 12 & $(12,9),(12,10)$ \\
                7 & $\lbrace11,12\rbrace$ & 11,12 & Coalesce & 13 & $(13,11),(13,12)$ \\
                8 & $\lbrace13\rbrace$ & - & - & - & - \\ \bottomrule
                \end{tabular}
\end{table}

  
\begin{figure}[htbp]
    \begin{center}
    \includegraphics[width=1.0\textwidth]{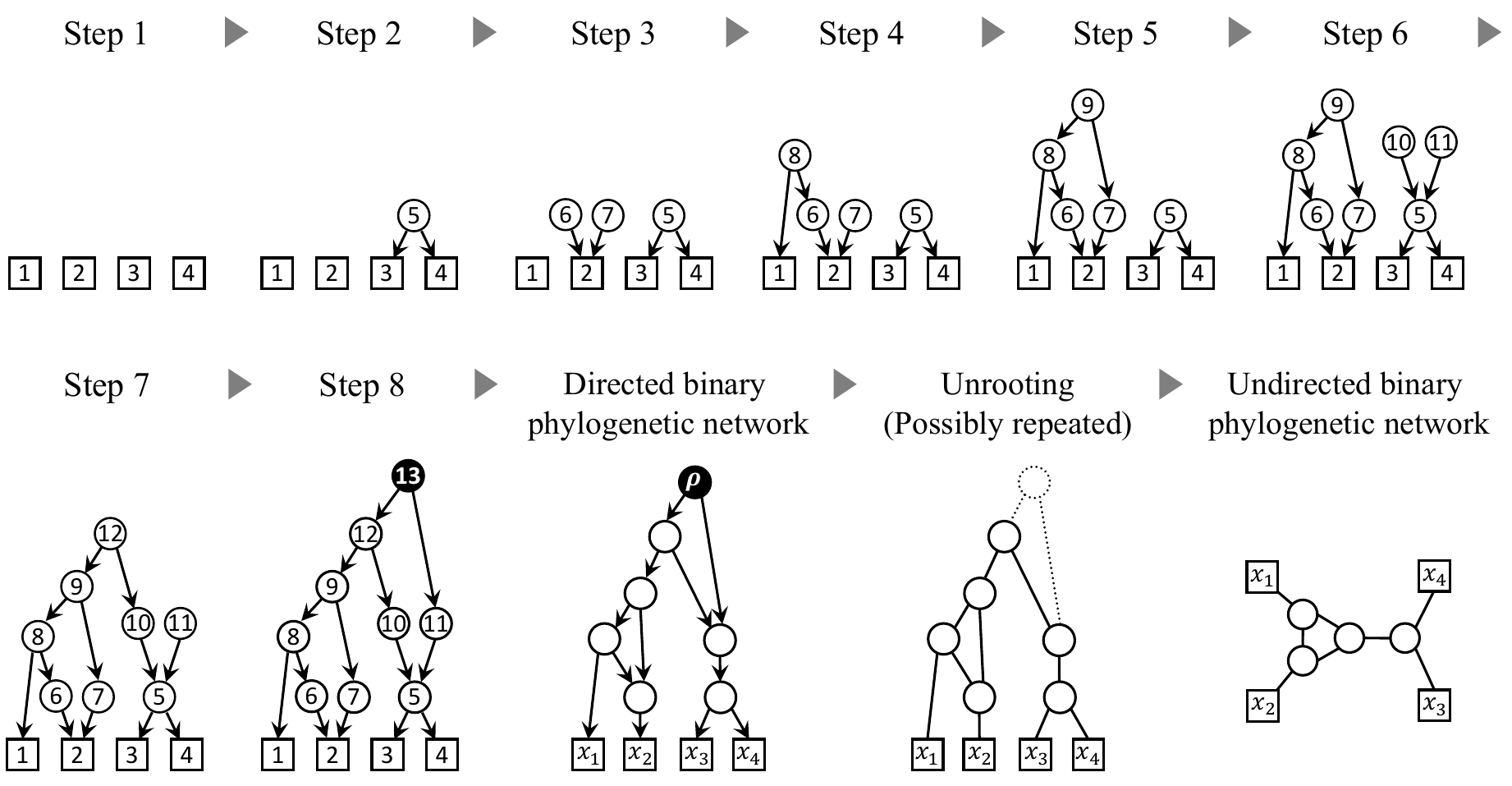}
    \caption{
        An illustration of how the process in Table \ref{table:graph.generate} yields a directed binary phylogenetic network on $X=\{x_1,x_2,x_3,x_4\}$ and how it is converted to an undirected one. 
\label{fig:graph.generate}}
    \end{center}
\end{figure}

\begin{table}[htbp]
    \caption{
        Summary of the undirected binary phylogenetic networks on $X$ generated before the experiments. 
            \label{table:materials}
        }
    \begin{tabular}{@{}lrrrrrrrrrrrr@{}}
    \toprule                                         & \multicolumn{11}{c}{Number $r$ of reticulations}                               &       \\
    $(|X|, P_r)$ & $0$   & $1$   & $2$   & $3$   & $4$   & $5$  & $6$  & $7$  & $8$  & $9$ & $10 \leq$ & Total \\ \midrule
    (10, 0.05)                           & 143 & 48  & 8   & 1   & 0   & 0  & 0  & 0  & 0  & 0 & 0 & 200   \\
    (10, 0.1)                             & 99  & 67  & 20  & 9   & 3   & 2  & 0  & 0  & 0  & 0 & 0 & 200   \\
    (10, 0.15)                            & 66  & 55  & 29  & 21  & 20  & 6  & 1  & 0  & 1  & 0 & 1 & 200   \\
    (20, 0.05)                           & 73  & 82  & 31  & 10  & 3   & 0  & 0  & 0  & 1  & 0 & 0 & 200   \\
    (20, 0.1)                             & 37  & 62  & 40  & 29  & 16  & 5  & 2  & 6  & 2  & 1 & 0 & 200   \\
    (20, 0.15)                            & 14  & 19  & 37  & 38  & 30  & 18 & 18 & 10 & 7  & 4 & 5 & 200   \\\midrule
    Total                                & 432 & 333 & 165 & 108 & 72 & 31 & 21 & 16 & 11 & 5 & 6 & 1200  \\ \bottomrule
    \end{tabular}
\end{table}



\end{appendices}

\backmatter

\bmhead{Supplementary information}
The source code, datasets and detailed experimental results are available at 
\url{https://github.com/hayamizu-lab/tree-child-orienter}. 
 (archived at swh:1:dir:666a10c01c14741702fddd5f8704b30bc90299e5).
\bmhead{Acknowledgements}
MH is supported by JST FOREST Program Grant Number JPMJFR2135, Japan. We thank Yufeng Wu and Louxin Zhang for sharing code for randomly generating graphs, which we adapted for our experiments. We also appreciate the anonymous reviewers’ careful reading of our manuscript and useful comments.

\bibliography{sn-article}

\end{document}